\documentclass[final,1p,times,authoryear]{elsarticle}
\usepackage{amsmath}
\usepackage{amsthm}
\usepackage{amssymb}
\usepackage{amsfonts}
\usepackage{dsfont}

\usepackage{float}
\usepackage[plainpages=false,pdfpagelabels=true,colorlinks=true,citecolor=blue,hypertexnames=false]{hyperref}
\usepackage[ruled,vlined]{algorithm2e}
\usepackage{mathtools}
\usepackage{extarrows}
\usepackage{url}
\usepackage{balance}
\usepackage{multirow}
\usepackage{color}
\usepackage{diagbox}
\usepackage{booktabs}
\usepackage{makecell}
\usepackage{hyperref}
\newcolumntype{V}{!{\vrule width 1pt}}
\newtheorem{theorem}{Theorem}
\newtheorem{lemma}[theorem]{Lemma}
\newtheorem{corollary}[theorem]{Corollary}

\newtheorem{proposition}[theorem]{Proposition}
\newtheorem{definition}[theorem]{Definition}

\newtheorem{example}[theorem]{Example}
\newtheorem{remark}[theorem]{Remark}
\newtheorem{problem}[theorem]{Problem}

\def\lcm{{\rm lcm}}
\def\gcd{{\rm gcd}}
\def\diag{{\rm diag}}
\def\rank{{\rm rank}}

\begin{document}

\begin{frontmatter}

\title{On the equivalence problem of Smith forms for multivariate polynomial matrices}

\author[sju]{Dong Lu}
\ead{donglu@swjtu.edu.cn}

\author[klmm,ucas]{Dingkang Wang}
\ead{dwang@mmrc.iss.ac.cn}

\author[hunu]{Fanghui Xiao}
\ead{xiaofanghui@hunnu.edu.cn}

\author[SHANT]{Xiaopeng Zheng\corref{cor}}
\ead{zhengxiaopeng@amss.ac.cn}

\cortext[cor]{Corresponding author}

\address[sju]{School of Mathematics, Southwest Jiaotong University, Chengdu 610031, China}

\address[klmm]{KLMM, Academy of Mathematics and Systems Science, CAS, Beijing 100190, China}

\address[ucas]{School of Mathematical Sciences, University of Chinese Academy of Sciences, Beijing 100049, China}

\address[hunu]{MOE-LCSM, School of Mathematics and Statistics, Hunan Normal University, Changsha 410081, China}

\address[SHANT]{School of Mathematics and Computer Sciences, Shantou University, Shantou 515821, China}

\begin{abstract}
 This paper delves into the equivalence problem of Smith forms for multivariate polynomial matrices. Generally speaking, multivariate ($n \geq 2$) polynomial matrices and their Smith forms may not be equivalent. However, under certain specific condition, we derive the necessary and sufficient condition for their equivalence. Let $F\in K[x_1,\ldots,x_n]^{l\times m}$ be of rank $r$, $d_r(F)\in K[x_1]$ be the greatest common divisor of all the $r\times r$ minors of $F$, where $K$ is a field, $x_1,\ldots,x_n$ are variables and $1 \leq r \leq \min\{l,m\}$. Our key findings reveal the result: $F$ is equivalent to its Smith form if and only if all the $i\times i$ reduced minors of $F$ generate $K[x_1,\ldots,x_n]$ for $i=1,\ldots,r$.
\end{abstract}

\begin{keyword}
 Smith forms, Multivariate polynomial matrices, Matrix equivalence, Reduced minors, Localization
\end{keyword}
\end{frontmatter}

\section{Introduction}

 An important issue in polynomial matrix theory is the equivalence problem which consists in testing whether two polynomial matrices are equivalent.  The importance of the polynomial matrix equivalence problem not only lies in the demands of mathematical theoretical research, but also in its extensive applications in various engineering fields such as multidimensional system theory and signal processing (see \citep{Bose1982,Bose2003} and the references therein). Due to the simplicity of the Smith form of a polynomial matrix and the fact that it contains all the information of the original matrix, mathematicians and engineers have been paying attention to the equivalence problem between polynomial matrices and their Smith forms for decades.

 When studying eigenvalues and eigenvectors of matrices in linear algebra, a well-known conclusion is that any $\lambda$-matrix (all the entries are univariate polynomials in $K[\lambda]$) is equivalent to its Smith form \citep{Kailath1980}. This equivalence is rooted in the fact that $K[\lambda]$ is a principal ideal domain, which boasts efficient computational properties, including Euclidean division, facilitating the transformation of $\lambda$-matrices into their Smith forms. However, this  equivalence may not necessarily hold in the multivariate case (where the number of variables is greater than or equal to $2$). Although some progress has been made in recent years \citep{Boudellioua2013Further,
 Boudellioua2010Serre,Cluzeau2015A,LiD2022The,Li2017On,Lin2006On,Liu2024,Lu2023New}, determining the equivalence between multivariate polynomial matrices and their corresponding Smith forms remains an unresolved challenge.

 \cite{Frost1978} conducted the first study on the equivalence relation between any bivariate polynomial matrix $F\in K[x_1,x_2]^{l\times l}$ and its Smith form, asserting that the necessary and sufficient condition for their equivalence is when all the $i\times i$ reduced minors of $F$ generate $K[x_1,x_2]$, where $K$ is a field, $x_1,x_2$ are variables and $i=1,\ldots,l$. However, this assertion is not valid, as they have provided a counter-example in \citep{Frost1981}. In fact, under general circumstances, the above condition is merely a necessary one. Therefore, people have begun to consider under what conditions the matrix must satisfy for the assertion to hold true.

 \cite{Lee1983} transformed the above equivalence problem into solving a system of equations. However, although this method is theoretically feasible, it is not suitable from a computational perspective. \cite{LiD2019} focused on the equivalence between a certain type of bivariate polynomial matrices and their Smith forms. Let $F\in K[x_1,x_2]^{l\times l}$ be of full rank, and $\det(F)=p^t$ with $p\in K[x_1]$ being an irreducible polynomial, where $t$ is a positive integer. They proved that $F$ is equivalent to the diagonal matrix $\diag\{1,\ldots,1,p^t\}$ if and only if all the $(l-1)\times (l-1)$ minors of $F$ generate $K[x_1,x_2]$. Clearly, $\diag\{1,\ldots,1,p^t\}$ is the most special form among all possible Smith forms of $F$. There is no doubt that the general form of the Smith form of $F$ is $\diag\{p^{r_1},p^{r_2},\ldots,p^{r_l}\}$, where $r_i\mid r_{i+1}$ for $i=1,\ldots,l-1$, and $\sum_{j=1}^{l}r_j = t$. \cite{Zheng2023New} proved that Frost and Storey's assertion is correct in this case. Subsequently, \cite{Guan2024,Guan2024Further} used the Suslin's stability theorem and a generalized theorem of Vaserstein to extend the main result in \citep{Zheng2023New} to the multivariate case, i.e., if $F\in K[x_1,\ldots,x_n]^{l\times l}$ and $\det(F)=p^t$ with $p\in K[x_1]$ being an irreducible polynomial, then $F$ is equivalent to $\diag\{p^{r_1},p^{r_2},\ldots,p^{r_l}\}$ if and only if all the $i\times i$ reduced minors of $F$ generate $K[x_1,\ldots,x_n]$, where $i=1,\ldots,l$.

 The main objective of this paper is to expand the scope of limitation ``$\det(F)=p^t$ with $p\in K[x_1]$ being an irreducible polynomial'' to $\det(F)\in K[x_1]$, and explore the necessary and sufficient condition for $F$ to be equivalent to its Smith form under the new constraint.

 Assuming that $\det(F)=p_1^{r_1}p_2^{r_2}\cdots p_t^{r_t}$ is an irreducible decomposition of $\det(F)$, where $p_1,p_2,$ $\ldots,p_t\in K[x_1]$ are distinct and irreducible polynomials, and $r_1,r_2,\ldots,r_t$ are positive integers. When $n=2$ and $t=2$, \cite{Lu2023} proved that $F$ is equivalent to $\diag\{1,\ldots,1,p_1^{t_1}p_2^{t_2}\}$ if and only if all the $(l-1)\times (l-1)$ minors of $F$ generate $K[x_1,x_2]$. The idea proposed in \citep{Lu2023} is to establish a homomorphism mapping from $K[x_1,x_2]$ to the quotient ring $(K[x_1]/\langle p_1p_2\rangle)[x_2]$, then reduce $F$ to its Smith form by using polynomial division in $(K[x_1]/\langle p_1p_2\rangle)[x_2]$. However, some examples have showed that this method is not applicable when $n \geq 3$ or $t \geq 3$. Consequently, the employment of an additional localization technique becomes imperative to tackle the issue effectively. In this paper, we strive to completely solve the equivalence problem between this class of multivariate polynomial matrices and their Smith forms by utilizing the localization technique, ultimately arriving at the following main result.

 \begin{theorem}[{$=$ Theorem \ref{main-theorem-1} in Section \ref{sec_Biv}}]\label{sec-1-theorem-1}
  Let $F\in K[x_1,\ldots,x_n]^{l\times l}$ and $\det(F)\in K[x_1]$, then $F$ is equivalent to its Smith form if and only if all the $i\times i$ reduced minors of $F$ generate $K[x_1,\ldots,x_n]$, where $i=1,\ldots,l$.
 \end{theorem}

 Subsequently, we generalize the above theorem to the case of non-square and non-full rank matrices with the aid of the Quillen-Suslin theorem, and thereby obtain the following corollary.

 \begin{corollary}[{$=$ Corollary \ref{main-corollary} in Section \ref{sec_Biv}}]\label{sec-1-corollary-1}
  Let $F\in K[x_1,\ldots,x_n]^{l\times m}$ with rank $r$ and $d_r(F)\in K[x_1]$, where $1 \leq r \leq \min\{l,m\}$. Then $F$ is equivalent to its Smith form if and only if all the $i\times i$ reduced minors of $F$ generate $K[x_1,\ldots,x_n]$, where $i=1,\ldots,r$.
 \end{corollary}

  The rest of the paper is organized as follows. In Section \ref{sec_Pre}, we recall some terminologies and some preliminary results needed in this paper. In addition, we present the major problem that shall consider. In Section \ref{sec_Biv}, we first use an example to illustrate the challenges faced by the problem to be solved, then tackle them in turn, and finally provide the proofs of Theorem \ref{sec-1-theorem-1} and Corollary \ref{sec-1-corollary-1}. Some concluding remarks are provided in Section \ref{sec_conclusions}.

\section{Preliminaries}\label{sec_Pre}

 Let $K$ be a field, $K[x_1,\ldots,x_n]$ be the polynomial ring in variables $x_1,\ldots,x_n$ over $K$, and $K[x_1,\ldots,x_n]^{l\times m}$ be the set of $l\times m$ matrices with entries in $K[x_1,\ldots,x_n]$. Without loss of generality, we assume that $l \leq m$ throughout the paper. In addition, we use ``w.r.t.'' to represent ``with respect to''.

 For any given $F\in K[x_1,\ldots,x_n]^{l\times m}$, let $\rank(F)$ be the rank of $F$. We use $d_i(F)$ to denote the greatest common divisor of all the $i\times i$ minors of $F$, where $i=1,\ldots,\rank(F)$. Here, we make the convention that $d_0(F) \equiv 1$ and $d_i(F) \equiv 0$ for all $i > \rank(F)$. For convenience, let $\diag\{f_1,\ldots,f_l\}$ represent the $l\times l$ diagonal matrix whose diagonal elements are $f_1,\ldots,f_l$, where $f_1,\ldots,f_l\in K[x_1,\ldots,x_n]$.

 \subsection{Matrix equivalence}

 As we all know, $U\in K^{l\times l}$ is called invertible in linear algebra if and only if $\det(U) \neq 0$. Similarly, in fields such as commutative algebra and multidimensional systems, there exists a comparable definition, detailed below.

 \begin{definition}
  Let $R$ be a commutative ring and $U\in R^{l\times l}$, then $U$ is said to be {\em unimodular} if $\det(U)$ is a unit in $R$.
 \end{definition}

 According to the above definition, if $R = K[x_1,\ldots,x_n]$, then $U$ is a unimodular matrix if and only if $\det(U) \in K \setminus \{0\}$. The following definition occupies an important position in the main results of the paper.

\begin{definition}[{\cite{Lin1988}}]
 Let $F \in K[x_1,\ldots,x_n]^{l \times m}$ with rank $r$, where $1\leq r \leq l$. For any given integer $i$ with $1 \leq i \leq r$, let $a^{(i)}_{1}, \ldots, a^{(i)}_{\beta_i}$ be all the $i \times i$ minors of $F$, where $\beta_i={l \choose i}{m \choose i}$. Extracting $d_{i}(F)$ from $a^{(i)}_{1}, \ldots, a^{(i)}_{\beta_i}$ yields
 \[ a^{(i)}_{j}=d_{i}(F) \cdot b^{(i)}_{j}, ~ j=1, \ldots, \beta_i.\]
 Then, $b^{(i)}_{1}, \ldots, b^{(i)}_{\beta_i}$ are called all the $i \times i$ {\em reduced minors} of $F$.
\end{definition}

 For convenience, we use $J_i(F)$ to denote the ideal in $K[x_1,\ldots,x_n]$ generated by all the $i\times i$ reduced minors of $F$.

 Now, we recall the concept of matrix equivalence.

 \begin{definition}
  Let $R$ be a commutative ring, and $F,Q\in R^{l\times m}$. Then $F$ is said to be equivalent to $Q$ if there are two unimodular matrices $U\in  R^{l\times l}$ and $V\in R^{m\times m}$ such that $UFV=Q$. For convenience, $F$ being equivalent to $Q$ over $R$ is denoted by $F \sim_R Q$.
 \end{definition}

 For any two equivalent matrices $A$ and $B$, they possess the following desirable property.

 \begin{proposition}[{\cite{Zheng2023New}}]\label{lemma-reduced-2}
  Let $A,B\in K[x_1,\ldots,x_n]^{l\times m}$. If $A \sim_{K[x_1,\ldots,x_n]} B$, then $d_i(A) = d_i(B)$ and $J_i(A) = J_i(B)$, where $i=1,\ldots,l$.
 \end{proposition}

 The Smith form for any multivariate polynomial matrix is defined as follows.

\begin{definition}
 Let $F \in K[x_1,\ldots,x_n]^{l \times m}$ with rank $r$, and $\Phi_{i}$ be a
 polynomial defined as follows:
 \[\Phi_{i}=\left\{\begin{array}{ll}
	\frac{d_{i}(F)}{d_{i-1}(F)}, & 1 \leq i \leq r; \\
	~~~~0, & r<i \leq l.
	\end{array}\right.\]
 Moreover, $\Phi_{i}$ satisfies the divisibility property $\Phi_{1} \mid  \Phi_{2} \mid  \cdots  \mid \Phi_{r}$. Then the {\em Smith form} of $F$ is given by
 \[\begin{pmatrix}
     {\rm diag}\{\Phi_1,\ldots, \Phi_r\} &  0_{r \times (m-r)}  \\
         0_{(l-r)\times r}    &  0_{(l-r)\times (m-r)}
     \end{pmatrix}.\]
\end{definition}

 \subsection{Localization}

 In this paper, we mainly focus on the equivalence problem between multivariate polynomial matrices and their Smith forms. To study this problem, we need some additional concepts and related results.

 \begin{definition}[{\cite{Youla1979Notes}}]
  Let $F \in K[x_1,\ldots,x_n]^{l \times m}$ be of full rank, then $F$ is said to be zero left prime (ZLP) if all the $l \times l$ minors of $F$ generate $K[x_1,\ldots,x_n]$.
 \end{definition}

 The following theorem is closely related to ZLP matrices. It was first proposed by \cite{Lin2001A} as a conjecture, known as the Lin-Bose conjecture, and later solved by \cite{Liu2014The,Pommaret2001,SV2004,Wang2004On} respectively using different methods.

 \begin{theorem}[\cite{Lin2001A}]\label{Lin-Bose-theorem}
  Let $F\in K[x_1,\ldots,x_n]^{l\times m}$ with rank $r$, where $1\leq r \leq l$. If $J_r(F) = K[x_1,\ldots,x_n]$, then there exist two polynomial matrices $G_1\in K[x_1,\ldots,x_n]^{l\times r}$ and $F_1\in K[x_1,\ldots,x_n]^{r\times m}$ such that $F = G_1 F_1$ and $F_1$ is a ZLP matrix.
 \end{theorem}

 In 1955, \cite{Serre1955} proposed a famous conjecture, which is referred to as Serre's conjecture in \citep{Lam1978}, claiming that projective modules over polynomial rings are free. This conjecture was independently solved by \cite{Quillen1976Projective} and \cite{Suslin1976Projective}, and the result is called the Quillen-Suslin theorem.

 \begin{theorem}[{Quillen-Suslin theorem}]\label{QS-theorem}
  Let $F \in K[x_1,\ldots,x_n]^{l \times m}$ be a ZLP matrix with $l<m$, then a unimodular matrix $U \in K[x_1,\ldots,x_n]^{m \times m}$ can be constructed such that $F$ is its first $l$ rows.
 \end{theorem}

 Quillen and Suslin provided two existential proofs, which implies that Serre's conjecture is correct. However, the method to construct a unimodular matrix that meets the criterion stated in the Quillen-Suslin theorem remains unknown. It was not until 1992 that  \cite{Logar1992} first presented an algorithm for calculating the unimodular matrix in Theorem \ref{QS-theorem}, utilizing the key technology called localization proposed by \cite{Quillen1976Projective}. Now, let us introduce the concept of localization, which will play a significant role in the process of problem-solving in the paper.

\begin{definition}[{See Definition 1.4.4, Page 32, \cite{Greuel2008}}]\label{def_local}
 Let $R$ be a commutative ring with unity, where the unity is denoted by $1$.
\begin{enumerate}
  \item[(1)] A subset $S \subset R$ is called {\em multiplicatively closed} if $1\in S$ and $a,b\in S$ implies $ab\in S$.

  \item[(2)] Let $S \subset R$ be multiplicatively closed. We define the {\em localization} $R_S$ of $R$ w.r.t. $S$ as follows:
 \[R_S:=\left\{ \frac{f}{s} ~ \Big{|} ~ f\in R, s\in S\right\},\]
 where $f/s$ denotes the equivalence class of $(f,s)\in R\times S$ w.r.t. the following equivalence relation:
 \[ (f,s) \sim (f',s') \Longleftrightarrow \exists ~ s''\in S \text{ such that } s''(fs'-f's)=0.\]
\end{enumerate}
\end{definition}

\begin{remark}
 If $R$ in Definition \ref{def_local} is an integral domain (i.e., there are no zero divisors in $R$), then $(f,s) \sim (f',s')$ implies $fs'-f's=0$. It is straightforward to confirm that performing sums and products on $R_S$ using the same way as for ordinary fractions gives well-defined operations. Furthermore, under these sum and product operations, $R_S$ is a new commutative ring with unity that contains $R$.
\end{remark}

\begin{definition}
 Let $R$ be a commutative ring with unity, and $S_1,\ldots,S_r\subset R$ be multiplicatively closed. If for any $s_1\in S_1,\ldots,s_r\in S_r$, there exist $u_1,\ldots,u_r\in R$ such that $u_1s_1+\cdots+u_rs_r = 1$, then $S_1,\ldots,S_r$ are called {\em comaximal}.
\end{definition}

 \begin{example}
  Let $f_1,\ldots,f_r\in K[x_1,\ldots,x_n]\setminus \{0\}$ be distinct polynomials, and we construct the following sets
  \[ S_i:=\left\{f_i^t \mid t\in \mathbb{N}\right\}, ~~ i=1,\ldots,r,\]
  where $\mathbb{N}$ is the set of natural numbers that includes $0$. It is easy to check that $S_i$ is multiplicatively closed, where $i=1,\ldots,r$. In addition, $S_1,\ldots,S_r$ are comaximal if and only if $\langle f_1,\ldots,f_r \rangle = 1$ (i.e., the ideal generated by $f_1,\ldots,f_r$ is the unit ideal $K[x_1,\ldots,x_n]$).
 \end{example}

 \begin{lemma}[See Lemma 43, Page 28, {\cite{Yengui2015Con}}]\label{main-lem-yengui}
  Let $F\in R[y]^{l\times m}$, where $R$ is a commutative ring with unity and $y$ is a variable. Let  $S\subset R$ be multiplicatively closed, and $z$ be a new variable. If $F(y) \sim_{R_S[y]} F(0)$, then there is a polynomial $s\in S$ such that $F(y+sz) \sim_{R[y,z]} F(y)$.
 \end{lemma}

 To make Lemma \ref{main-lem-yengui} applicable to this paper, we modify it and obtain the following result.

 \begin{lemma}\label{main-lem-1-multi}
  Let $F\in K[x_1,\ldots,x_n]^{l\times l}$, $S\subset K[x_1]$ is multiplicatively closed and $z$ is a new variable. If $F(x_1,\ldots,x_{n-1},x_n) \sim_{(K[x_1]_S)[x_2,\ldots,x_n]} F(x_1,\ldots,x_{n-1},0)$, then there exists a polynomial $s\in S$ such that
  \[ F(x_1,\ldots,x_{n-1},x_n+sz) \sim_{K[x_1,\ldots,x_{n-1},x_n,z]} F(x_1,\ldots,x_{n-1},x_n).\]
 \end{lemma}

 Lemma \ref{main-lem-1-multi} is essentially the same as that of Lemma \ref{main-lem-yengui}, except that we have substituted $(K[x_1]_S)[x_2,\ldots,x_n]$ for $R_S[y]$. Please see \ref{sec:appendix} for the detailed proof process.

 \subsection{Problem}

 In this paper, we focus on the following problem.

 \begin{problem}\label{mian-problem}
  Let $F\in K[x_1,\ldots,x_n]^{l\times m}$ with rank $r$ and $d_r(F)\in K[x_1]$, where $1 \leq r \leq l$. What is the necessary and sufficient condition for the equivalence of $F$ and its Smith form?
 \end{problem}

 To solve problem \ref{mian-problem}, we also need to introduce the concept of homomorphic mapping and two related important lemmas.

 \begin{definition}
  Let $p\in K[x_1]$ be an irreducible polynomial, then $K[x_1]/\langle p \rangle$ is a filed. We consider the following homomorphism
  \[\begin{array}{cccc}
    \phi:  & K[x_1,\ldots,x_n] &  \longrightarrow & (K[x_1]/\langle p \rangle)[x_2,\ldots,x_n] \\
           & \sum c_{i_2\cdots i_n}(x_1)\cdot x_2^{i_2}\cdots x_n^{i_n} & \longrightarrow & \sum \overline{c_{i_2\cdots i_n}(x_1)} \cdot x_2^{i_2}\cdots x_n^{i_n}, \\
   \end{array}\]
  where $c_{i_2\cdots i_n}(x_1)\in K[x_1]$ and $\overline{c_{i_2\cdots i_n}(x_1)}\in K[x_1]/\langle p \rangle$. This homomorphism can extend canonically to the homomorphism $\phi: K[x_1,\ldots,x_n]^{l\times l} \rightarrow (K[x_1]/\langle p \rangle)[x_2,\ldots,x_n]^{l\times l}$ by applying $\phi$ entry-wise.
 \end{definition}

 For convenience, we use $\overline{F}$ to denote the polynomial matrix $\phi(F)$ in $(K[x_1]/\langle p \rangle)[x_2,\ldots,x_n]^{l\times l}$.

 Based on the Suslin's stability theorem \citep{Park1995An,Suslin1977On} and a generalized theorem of Vaserstein \citep{Yengui2015Con}, \cite{Guan2024,Guan2024Further} extended the bivariate case in \citep{Zheng2023New} to the case of $n \geq 2$, and obtained the following results.

\begin{lemma}[{\cite{Guan2024Further}}]\label{Zheng-lemma}
 Let $F\in K[x_1,\ldots,x_n]^{l\times l}$, and $p\in K[x_1]$ be an irreducible polynomial. If $J_k(F)=K[x_1,\ldots,x_n]$ and $\rank(\overline{F}) = k$, then
 \[ F \sim_{K[x_1,\ldots,x_n]} \diag\{\underbrace{1,\ldots,1}_{k},p,\ldots,p\} \cdot G\]
 for some matrix $G\in K[x_1,\ldots,x_n]^{l\times l}$.
\end{lemma}

\begin{lemma}[{\cite{Guan2024Further}}]\label{zheng-theorem}
 Let $B = \diag\{p^{s_1},\ldots,p^{s_k},p^s,\ldots,p^s\} \cdot U \cdot \diag\{\underbrace{1,\ldots,1}_{k},p,\ldots,p\}$, where $p\in K[x_1]$ is an irreducible polynomial, $s_1 \leq \cdots \leq s_k \leq s$ and $U \in K[x_1,\ldots,x_n]^{l\times l}$ is a unimodular matrix. If $d_i(B) = p^{s_1+\cdots+s_i}$ and $J_i(B) = K[x_1,\ldots,x_n]$ for $i=1,\ldots,k$, then
 \[ B \sim_{K[x_1,\ldots,x_n]} \diag\{p^{s_1},\ldots,p^{s_k},p^{s+1},\ldots,p^{s+1}\}.\]
\end{lemma}

\section{Necessary and Sufficient Condition}\label{sec_Biv}

 When $F\in K[x_1,x_2]^{l\times l}$ and $\det(F) = p^t$ with $p\in K[x_1]$ being an irreducible polynomial in Problem \ref{mian-problem}, \cite{Zheng2023New} first proved that the necessary and sufficient condition for this special case is that $J_i(F) = K[x_1,x_2]$ for $i=1,\ldots,l$. Subsequently, \cite{Guan2024,Guan2024Further} extended the above result to the case with $n$ variables. Based on the above-mentioned works, we assert that the necessary and sufficient condition for Problem \ref{mian-problem} is that $J_i(F) = K[x_1,\ldots,x_n]$ for $i=1,\ldots,r$. To prove the assertion, we first consider the case where $F$ is a square matrix with full rank, and then generalize it to the case where $F$ is a non-square and non-full rank matrix.

\subsection{Three challenges}

 Let $F\in K[x_1,\ldots,x_n]^{l\times l}$ and $\det(F)\in K[x_1]$. If $F$ and its Smith form are equivalent, it is easy to see that $J_i(F) = K[x_1,\ldots,x_n]$ for $i=1,\ldots,l$. However, proving the converse is considerably complex. Let us elucidate potential issues that may emerge through an illustrative example.

 \begin{example}\label{biv-example-idea}
  Let $F\in K[x_1,x_2,x_3]^{4\times 4}$, $\det(F) = p_1^6p_2^4p_3^2$, and $J_i(F) = K[x_1,x_2,x_3]$ for $i=1,\ldots,4$, where $p_1,p_2,p_3\in K[x_1]$ are distinct and irreducible polynomials. Without loss of generality, assume that the Smith form of $F$ is $\diag\{1,p_1p_2,p_1^2p_2p_3,p_1^3p_2^2p_3\}$.

  Since $J_1(F) = K[x_1,x_2,x_3]$ and $\rank(\overline{F}) = 1$, we apply Lemma \ref{Zheng-lemma} to extract $p_1$ from $F$ and obtain
 \begin{equation}\label{biv-idea-equ-1}
  F \sim_{K[x_1,x_2,x_3]} \diag\{1,p_1,p_1,p_1\} \cdot G_1,
 \end{equation}
 where $\overline{F}\in (K[x_1]/\langle p_1 \rangle)[x_2,x_3]^{4\times 4}$ and $G_1\in K[x_1,x_2,x_3]^{4\times 4}$. To continue extracting $p_1$ from $G_1$ using Lemma \ref{Zheng-lemma}, we must prove that $J_2(G_1) = K[x_1,x_2,x_3]$ and $\rank(\overline{G_1}) = 2$. This is the first challenge. Only by overcoming this challenge can we obtain
 \begin{equation}\label{biv-idea-equ-1-1}
  G_1 \sim_{K[x_1,x_2,x_3]} \diag\{1,1,p_1,p_1\} \cdot G_2,
 \end{equation}
 where $G_2\in K[x_1,x_2,x_3]^{4\times 4}$. Combining Equations \eqref{biv-idea-equ-1} and \eqref{biv-idea-equ-1-1}, we have
 \begin{equation}\label{biv-idea-equ-1-2}
  F \sim_{K[x_1,x_2,x_3]} \diag\{1,p_1,p_1,p_1\} \cdot U_1 \cdot \diag\{1,1,p_1,p_1\} \cdot G_2,
 \end{equation}
 where $U_1\in K[x_1,x_2,x_3]^{4\times 4}$ is a unimodular matrix. Let $B = \diag\{1,p_1,p_1,p_1\} \cdot U_1 \cdot \diag\{1,1,p_1,$ $p_1\}$. At this juncture, if $d_1(B) = 1$, $d_2(B) = p_1$ and $J_i(B) = K[x_1,x_2,x_3]$ for $i=1,2$, then it follows from
 Lemma \ref{zheng-theorem} that
  \begin{equation}\label{biv-idea-equ-2}
   B \sim_{K[x_1,x_2,x_3]} \diag\{1,p_1,p_1^2,p_1^2\}.
  \end{equation}
  Combining Equations \eqref{biv-idea-equ-1-2} and \eqref{biv-idea-equ-2}, we have
  \begin{equation}\label{biv-idea-equ-3}
   F \sim_{K[x_1,x_2,x_3]} \diag\{1,p_1,p_1^2,p_1^2\} \cdot G_3,
  \end{equation}
  where $G_3\in k[x_1,x_2,x_3]^{4\times 4}$. Now, we repeat the above computational process, extracting $p_1$ from $G_3$ and subsequently establishing the following equivalence relationship:
  \begin{equation}\label{biv-idea-equ-3-0}
   F \sim_{K[x_1,x_2,x_3]} \diag\{1,p_1,p_1^2,p_1^3\} \cdot G_4,
  \end{equation}
  where $G_4\in k[x_1,x_2,x_3]^{4\times 4}$.

  Let $D = \diag\{1,p_1,p_1^2,p_1^3\}$, then the second challenge is to prove that $d_i(F) = d_i(D) \cdot d_i(G_4)$ and $J_i(G_4) = K[x_1,x_2,x_3]$ for $i=1,\ldots,4$. Once we have solved this problem, we can repeat the above process to obtain
  \begin{equation}\label{biv-idea-equ-3-1}
   F \sim_{K[x_1,x_2,x_3]} \diag\{1,p_1,p_1^2,p_1^3\} \cdot V_1 \cdot \diag\{1,p_2,p_2,p_2^2\} \cdot V_2 \cdot \diag\{1,1,p_3,p_3\},
  \end{equation}
  where $V_1,V_2\in k[x_1,x_2,x_3]^{4\times 4}$ are two unimodular matrices.

  Finally, the third and most crucial challenge is to prove the equivalence between \[\diag\{1,p_1,p_1^2,p_1^3\} \cdot V_1 \cdot \diag\{1,p_2,p_2,p_2^2\} \text{ and } \diag\{1,p_1p_2,p_1^2p_2,p_1^3p_2^2\}.\]
  If we can prove, then by repeatedly applying this method to Equation \eqref{biv-idea-equ-3-1}, and we can deduce the result:
  \begin{equation}\label{biv-idea-equ-3-2}
   F \sim_{K[x_1,x_2,x_3]} \diag\{1,p_1p_2,p_1^2p_2p_3,p_1^3p_2^2p_3\}.
  \end{equation}
 \end{example}

 Below, we will sequentially address the three challenges presented in Example \ref{biv-example-idea}.

 \subsubsection{Solving the first challenge}

 Before we solve the first challenge, we introduce an important concept.

 \begin{definition}
  Suppose $F\in K[x_1,\ldots,x_n]^{l\times l}$ and $\det(F)\in K[x_1]$. Let  $p_1^{r_1}p_2^{r_2}\cdots p_t^{r_t}$ be the irreducible decomposition of $\det(F)$, where $p_1,p_2,\ldots,p_t\in K[x_1]$ are distinct and irreducible polynomials, and $r_1,r_2,\ldots,r_t$ are positive integers. Assume that the Smith form of $F$ is
  \begin{equation*}
   \left(
         \begin{array}{cccc}
           p_1^{s_{11}}p_2^{s_{21}}\cdots p_t^{s_{t1}} & &  & \\  &  p_1^{s_{12}}p_2^{s_{22}}\cdots p_t^{s_{t2}} &  & \\
           &  & \ddots & \\
            &  &  & p_1^{s_{1l}}p_2^{s_{2l}}\cdots p_t^{s_{tl}} \\
         \end{array}
       \right),
  \end{equation*}
  where $s_{ij}\in \mathbb{N}$, $s_{i1} \leq s_{i2} \leq \cdots \leq s_{il}$ and $\sum_{j=1}^{l}s_{ij} = r_i$ for $i=1,\ldots,t$. Then for each integer $i$, the diagonal matrix $\diag\{p_i^{s_{i1}},p_i^{s_{i2}},\ldots,p_i^{s_{il}}\}$ is called the Smith form of $F$ w.r.t. $p_i$.
 \end{definition}

 \begin{lemma}\label{Lemma_chal_1}
  Let $F\in K[x_1,\ldots,x_n]^{l\times l}$ with $\det(F)\in K[x_1]$, $J_i(F) = K[x_1,\ldots,x_n]$ for $i=1,\ldots,l$. Assume that $p_1\in K[x_1]$ is an irreducible polynomial such that $p_1\mid \det(F)$ and $\diag\{p_1^{s_{11}},p_1^{s_{12}},\ldots,p_1^{s_{1l}}\}$ is the Smith form of $F$ w.r.t. $p_1$. If there is an integer $s$ with $s_{1k} \leq s < s_{1(k+1)}$ such that
  \[ F \sim_{K[x_1,\ldots,x_n]} \diag\{p_1^{s_{11}},\ldots,p_1^{s_{1k}},p_1^{s},\ldots,p_1^{s}\} \cdot G_1\]
  for some matrix $G_1\in K[x_1,\ldots,x_n]^{l\times l}$, then we have
  \begin{enumerate}
    \item[(1)] $\gcd(p_1,d_i(G_1))=1$ and $J_i(G_1) = K[x_1,\ldots,x_n]$ for $i=1,\ldots,k$;

    \item[(2)] $\rank(\overline{G_1}) = k$, where $\overline{G_1} \in (K[x_1]/\langle p_1 \rangle)[x_2,\ldots,x_n]^{l\times l}$;

    \item[(3)] there exist a unimodular matrix $U\in K[x_1,\ldots,x_n]^{l\times l}$ and a matrix $G_2\in K[x_1,\ldots,x_n]^{l\times l}$ such that
        \[ F \sim_{K[x_1,\ldots,x_n]} \diag\{p_1^{s_{11}},\ldots,p_1^{s_{1k}},p_1^s,\ldots,p_1^s\} \cdot U \cdot \diag\{\underbrace{1,\ldots,1}_{k},p_1,\ldots,p_1\} \cdot G_2.\]
        In addition, let $B = \diag\{p_1^{s_{11}},\ldots,p_1^{s_{1k}},p_1^s,\ldots,p_1^s\} \cdot U \cdot \diag\{1,\ldots,1,p_1,\ldots,p_1\}$, then $d_i(B) = p_1^{s_{11}+\cdots+s_{1i}}$ and $J_i(B) = K[x_1,\ldots,x_n]$ for $=1,\ldots,k$.
  \end{enumerate}
 \end{lemma}

 \begin{remark}
  Lemma 3.5 in \citep{Zheng2023New} is a special case of the above lemma. Nonetheless, the previously employed proof method falls short in its applicability to this specific lemma. Hence, we introduce a new approach to validate it.
 \end{remark}

 \begin{proof}[Proof of Lemma {\rm \ref{Lemma_chal_1}}]
  Let $\sum_{j=1}^{i}s_{1j} = e_i$, where $i=1,\ldots,k$. Since $\diag\{p_1^{s_{11}},p_1^{s_{12}},\ldots,p_1^{s_{1l}}\}$ is the Smith form of $F$ w.r.t. $p_1$, we have
  \begin{equation}\label{Lu_equ-1}
   d_i(F) = p_1^{e_i}\theta_i,
  \end{equation}
  where $\theta_i\in K[x_1]$ satisfies that $\gcd(p_1,\theta_i) = 1$ for $i=1,\ldots,k$.

  (1) Let $A = \diag\{p_1^{s_{1}},\ldots,p_1^{s_{k}},p_1^{s},
  \ldots,p_1^{s}\} \cdot G_1$, then by the fact that $F\sim_{K[x_1,\ldots,x_n]} A$ we obtain
  \begin{equation}\label{Lu_equ-2}
   d_i(A) = d_i(F) \text{ and } J_i(A) = J_i(F) \text { for } i=1,\ldots,k.
  \end{equation}
  For any given integer $i$ with $1 \leq i \leq k$, assume that $h_{11},\ldots,h_{1\eta},h_{21},\ldots,h_{2\xi}\in K[x_1,\ldots,x_n]$ are all the $i\times i$ minors of $G_1$, where $h_{11},\ldots,h_{1\eta}$ are all the $i\times i$ minors of the submatrix formed by the first $i$ rows of $G_1$. Then,
  \[ p_1^{e_i}h_{11},\ldots,p_1^{e_i}h_{1\eta},
  p_1^{e_{21}}h_{21},\ldots,p_1^{e_{2\xi}}h_{2\xi}\]
  are all the $i\times i$ minors of $A$, where $e_{2j} \geq e_i$ for $j=1,\ldots,\xi$. Thus,
  \begin{equation}\label{Lu_equ-3}
   \begin{split}
   d_i(A) & = \gcd(p_1^{e_i}h_{11},\ldots,p_1^{e_i}h_{1\eta},
  p_1^{e_{21}}h_{21},\ldots,p_1^{e_{2\xi}}h_{2\xi}) \\ & =  p_1^{e_i} \cdot \gcd(h_{11},\ldots,h_{1\eta},
  p_1^{e_{21}-e_i}h_{21},\ldots,p_1^{e_{2\xi}-e_i}h_{2\xi}).
   \end{split}
  \end{equation}
  Combining Equations \eqref{Lu_equ-1}, \eqref{Lu_equ-2} and \eqref{Lu_equ-3}, we have
  \begin{equation}\label{Lu_equ-4}
   \theta_i = \gcd(h_{11},\ldots,h_{1\eta},
  p_1^{e_{21}-e_i}h_{21},\ldots,p_1^{e_{2\xi}-e_i}h_{2\xi}).
  \end{equation}
  It follows from $\theta_i \mid p_1^{e_{2j}-e_i}h_{2j}$ and $\gcd(p_1,\theta_i) = 1$ that $\theta_i \mid h_{2j}$, where $j=1,\ldots,\xi$. This implies that
  \begin{equation}\label{Lu_equ-5}
   \theta_i \mid \gcd(h_{11},\ldots,h_{1\eta},h_{21},\ldots,h_{2\xi}).
  \end{equation}
  It is obvious that $\gcd(h_{11},\ldots,h_{1\eta},h_{21},\ldots,h_{2\xi}) \mid \gcd(h_{11},\ldots,h_{1\eta},
  p_1^{e_{21}-e_i}h_{21},\ldots,p_1^{e_{2\xi}-e_i}h_{2\xi})$, then we have
  \begin{equation}\label{Lu_equ-5-1}
   \theta_i = \gcd(h_{11},\ldots,h_{1\eta},h_{21},\ldots,h_{2\xi}).
  \end{equation}
  It follows from Equation \eqref{Lu_equ-5-1} that $d_i(G_1) = \theta_i$ for $i=1,\ldots,k$. Therefore, $\gcd(p_1,d_i(G_1))=1$ for $i=1,\ldots,k$.  It is easy to see that
  \[\frac{h_{11}}{\theta_i},\ldots,\frac{h_{1\eta}}{\theta_i},
  p_1^{e_{21}-e_i}\frac{h_{21}}{\theta_i},\ldots,
  p_1^{e_{2\xi}-e_i}\frac{h_{2\xi}}{\theta_i}\]
  are all the $i\times i$ reduced minors of $A$. As $J_i(A) = K[x_1,\ldots,x_n]$ for $i=1,\ldots,k$, we obtain
  \[ \left\langle \frac{h_{11}}{\theta_i},\ldots,\frac{h_{1\eta}}{\theta_i},
  \frac{h_{21}}{\theta_i},\ldots,\frac{h_{2\xi}}{\theta_i} \right\rangle = K[x_1,\ldots,x_n].\]
  Therefore, $J_i(G_1) = K[x_1,\ldots,x_n]$ for $i=1,\ldots,k$.

  (2) The proof of $\rank(\overline{G_1}) = k$ is almost the same as that of Lemma 3.7 in \citep{Zheng2023New}, and it is omitted here.

  (3) Since $J_k(G_1) = K[x_1,\ldots,x_n]$ and $\rank(\overline{G_1}) = k$, by Lemma \ref{Zheng-lemma} we have
  \begin{equation}\label{Lu_equ-6}
   G_1 \sim_{K[x_1,\ldots,x_n]} \diag\{\underbrace{1,\ldots,1}_{k},p_1,\ldots,p_1\} \cdot G_2,
  \end{equation}
  where $G_2\in K[x_1,\ldots,x_n]^{l\times l}$. Moreover, by using the proof method in (1) above, we can also derive that
  \begin{equation}\label{Lu_equ-7}
   d_i(G_2) = d_i(G_1) = \theta_i \text{ and } J_i(G_2) = J_i(G_1) = K[x_1,\ldots,x_n] \text { for } i=1,\ldots,k.
  \end{equation}
  As $F \sim_{K[x_1,\ldots,x_n]} \diag\{p_1^{s_{11}},\ldots,p_1^{s_{1k}},p_1^{s},\ldots,p_1^{s}\} \cdot G_1$, we get
  \[ F \sim_{K[x_1,\ldots,x_n]} \diag\{p_1^{s_{11}},\ldots,p_1^{s_{1k}},p_1^s,\ldots,p_1^s\} \cdot U \cdot \diag\{1,\ldots,1,p_1,\ldots,p_1\} \cdot G_2,\]
  where $U\in K[x_1,\ldots,x_n]^{l\times l}$ is a unimodular matrix. Let $B = \diag\{p_1^{s_{11}},\ldots,p_1^{s_{1k}},p_1^s,\ldots,p_1^s\} \cdot U \cdot \diag\{1,\ldots,1,p_1,\ldots,p_1\}$ and $D = BG_2$, then
  \begin{equation}\label{Lu_equ-8}
   d_i(D) = d_i(F) = p_1^{e_i} \theta_i \text{ and } J_i(D) = J_i(F) = K[x_1,\ldots,x_n] \text { for } i=1,\ldots,k.
  \end{equation}
  For any given integer $i$ with $1 \leq i \leq k$, assume that $h_1,\ldots,h_\beta$, $f_1,\ldots,f_\beta$ and $g_1,\ldots,g_\beta$ are all the $i\times i$ minors of $D,B$ and $G_2$ respectively, where $\beta = \binom{l}{i}^2$. According to the special form of $B$, we can easily conclude that $p_1^{e_i}\mid f_j$ for $j=1,\ldots,\beta$. Based on the Cauchy-Binet formula, for any given $i\times i$ minor $h_t$ of $D$, there are $i\times i$ minors $f_{t_1},\ldots,f_{t_{\gamma_t}}$ of $B$ and $i\times i$ minors $g_{t_1},\ldots,g_{t_{\gamma_t}}$ of $G_2$ such that
  \begin{equation}\label{Lu_equ-9}
   h_t = \sum_{j=1}^{\gamma_t} f_{t_j} \cdot g_{t_j}.
  \end{equation}
  Dividing both sides of Equation \eqref{Lu_equ-9} by $d_i(D)$, we get
  \begin{equation}\label{Lu_equ-10}
   \frac{h_t}{d_i(D)} = \sum_{j=1}^{\gamma_t} \frac{f_{t_j}}{p_1^{e_i}} \cdot \frac{g_{t_j}}{\theta_i}.
  \end{equation}
  It follows from Equation \eqref{Lu_equ-10} that
   \begin{equation}\label{Lu_equ-11}
   \left\langle \frac{h_1}{d_i(D)},\ldots,\frac{h_\beta}{d_i(D)} \right\rangle \subseteq
   \left\langle \frac{f_1}{p_1^{e_i}},\ldots,\frac{f_\beta}{p_1^{e_i}}\right\rangle.
  \end{equation}
  Since $\frac{h_1}{d_i(D)},\ldots,\frac{h_\beta}{d_i(D)}$ are all the $i\times i$ reduced minors of $D$, we obtain
  \begin{equation}\label{Lu_equ-12}
   \left\langle \frac{f_1}{p_1^{e_i}},\ldots,\frac{f_\beta}{p_1^{e_i}}\right\rangle = K[x_1,\ldots,x_n].
  \end{equation}
  As $f_1,\ldots,f_\beta$ are all the $i\times i$ minors of $B$, we get
  \[d_i(B) = p_1^{e_i} \text{ and } J_i(B) = K[x_1,\ldots,x_n] \text { for } i=1,\ldots,k.\]
  The proof is completed.
 \end{proof}

 \begin{remark}\label{remark_lu-1}
  Combining the above lemma with Lemma \ref{zheng-theorem}, we obtain
  \[F \sim_{K[x_1,\ldots,x_n]} \diag\{p_1^{s_{11}},\ldots,p_1^{s_{1k}},p_1^{s+1},\ldots,p_1^{s+1}\} \cdot G_3\]
  for some matrix $G_3\in K[x_1,\ldots,x_n]^{l\times l}$. If $s+1 < s_{1(k+1)}$, then we can continue to utilize Lemma \ref{Lemma_chal_1} to extract $p_1$ from $G_3$. If $s+1 = s_{1(k+1)}$, then we compare the sizes of $s+1$ and $s_{1(k+2)}$ and utilize Lemma \ref{Lemma_chal_1} to factorize $G_3$. By repeating the above operations a finite number of times, we can ultimately arrive at the following equivalence relation:
  \[F \sim_{K[x_1,\ldots,x_n]} \diag\{p_1^{s_{11}},p_1^{s_{12}},\ldots,p_1^{s_{1l}}\} \cdot G\]
  for some matrix $G\in K[x_1,\ldots,x_n]^{l\times l}$.
 \end{remark}

 Lemma \ref{Lemma_chal_1} solves the first challenge raised in Example \ref{biv-example-idea}.

 \subsubsection{Solving the second challenge}

 \begin{lemma}\label{lemma-reduced-1}
  Let $F,F_1,F_2\in K[x_1,\ldots,x_n]^{l\times l}$ satisfy $F=F_1F_2$ and $\gcd(\det(F_1),\det(F_2))=1$, then
  \begin{enumerate}
    \item[(1)] $d_i(F) = d_i(F_1) \cdot d_i(F_2)$, where $i=1,\ldots,l$;

    \item[(2)] if $J_i(F) = K[x_1,\ldots,x_n]$, then $J_i(F_1) = J_i(F_2) = K[x_1,\ldots,x_n]$, where $i=1,\ldots,l$.
  \end{enumerate}
 \end{lemma}

 \begin{proof}
  (1) Let $h_1^{(i)},\ldots,h_{\beta_i}^{(i)}$, $f_1^{(i)},\ldots,f_{\beta_i}^{(i)}$ and $g_1^{(i)},\ldots,g_{\beta_i}^{(i)}$ be all the $i\times i$ minors of $F,F_1$ and $F_2$ respectively, where $\beta_i = \binom{l}{i}^2$ and $i=1,\ldots,l$. For any given $i\times i$ minor $h_t^{(i)}$ of $F$, it follows from the Cauchy-Binet formula that there are $f_{t_1}^{(i)},\ldots,f_{t_{\eta_i}}^{(i)}$ and $g_{t_1}^{(i)},\ldots,g_{t_{\eta_i}}^{(i)}$ such that
  \begin{equation}\label{prop-equ-1}
   h_t^{(i)} = \sum_{j=1}^{\eta_i} f_{t_j}^{(i)} \cdot g_{t_j}^{(i)},
  \end{equation}
  where $\eta_i = \binom{l}{i}$. According to Equation \eqref{prop-equ-1}, we have
  \[ d_i(F_1) \mid h_t^{(i)} \text{ and } d_i(F_2) \mid h_t^{(i)},\]
  where $t=1,\ldots,\beta_i$. This implies that $d_i(F_1) \mid d_i(F)$ and $d_i(F_2) \mid d_i(F)$. Notice that $d_i(F_1) \mid \det(F_1)$, $d_i(F_2) \mid \det(F_2)$ and $\gcd(\det(F_1),\det(F_2))=1$, we have \[\gcd(d_i(F_1),d_i(F_2)) = 1 \text{ and } d_i(F_1) \cdot d_i(F_2) \mid d_i(F).\]
  Assume that $d_i(F) = \alpha_i \cdot d_i(F_1) \cdot d_i(F_2)$, where $\alpha_i\in K[x_1,\ldots,x_n]$. In the following we only need to prove that $\alpha_i =1$.

  Let $F_2^\ast\in K[x_1,\ldots,x_n]^{l\times l}$ be the adjoint matrix of $F_2$, then $F F_2^\ast = \det(F_2) \cdot F_1$. For any given $i\times i$ minor $f_s^{(i)}$ of $F_1$, it follows from the Cauchy-Binet formula that there are $h_{s_1}^{(i)},\ldots,h_{s_{\eta_i}}^{(i)}$ and $q_{s_1}^{(i)},\ldots,q_{s_{\eta_i}}^{(i)}$ such that
  \begin{equation}\label{prop-equ-2}
   (\det(F_2))^i \cdot f_s^{(i)} = \sum_{j=1}^{\eta_i} h_{s_j}^{(i)} \cdot q_{s_j}^{(i)},
  \end{equation}
  where $q_{s_1}^{(i)},\ldots,q_{s_{\eta_i}}^{(i)}$ are $i\times i$ minors of $F_2^\ast$, and $\eta_i = \binom{l}{i}$. It follows from Equation \eqref{prop-equ-2} that
  \[ d_i(F) \mid (\det(F_2))^i \cdot f_s^{(i)},\]
  where $s=1,\ldots,\beta_i$. As $d_i(F) = \alpha_i \cdot d_i(F_1)\cdot d_i(F_2)$ and $d_i(F_1)\mid f_s^{(i)}$, we get $\alpha_i \cdot d_i(F_2) \mid (\det(F_2))^i \cdot b_s^{(i)}$, where $b_s^{(i)}$ is the $i\times i$ reduced minor of $F_1$ w.r.t. $f_s^{(i)}$. It is obvious that
  \[ \alpha_i \mid (\det(F_2))^i \cdot b_s^{(i)}, ~ s=1,\ldots,\beta_i.\]
  Since $b_1^{(i)},\ldots,b_{\beta_i}^{(i)}$ are all the $i\times i$ reduced minor of $F_1$, we get
  \begin{equation}\label{prop-equ-3}
   \alpha_i \mid (\det(F_2))^i.
  \end{equation}
  Using the same argument as the above proof, we can easily get that
  \begin{equation}\label{prop-equ-4}
   \alpha_i \mid (\det(F_1))^i.
  \end{equation}
  Combining Equations \eqref{prop-equ-3} and \eqref{prop-equ-4}, we obtain
  \[\alpha_i \mid \gcd((\det(F_1))^i,(\det(F_2))^i).\]
  Due to $\gcd(\det(F_1),\det(F_2))=1$, it follows that $\alpha_i=1$, where $i=1,\ldots,l$.

  (2) Dividing both sides of Equation \eqref{prop-equ-1} by $d_i(F)$, we get
  \begin{equation}\label{prop-equ-5}
   \frac{h_t^{(i)}}{d_i(F)} = \sum_{j=1}^{\eta_i} \frac{f_{t_j}^{(i)}}{d_i(F_1)} \cdot \frac{g_{t_j}^{(i)}}{d_i(F_2)},
  \end{equation}
  where $\eta_i = \binom{l}{i}$, $t=1,\ldots,\beta_i$. Since $\frac{h_1^{(i)}}{d_i(F)},\ldots,\frac{h_{\beta_i}^{(i)}}{d_i(F)}$, $\frac{f_1^{(i)}}{d_i(F_1)},\ldots,\frac{f_{\beta_i}^{(i)}}{d_i(F_1)}$ and
  $\frac{g_1^{(i)}}{d_i(F_2)},\ldots,\frac{g_{\beta_i}^{(i)}}{d_i(F_2)}$ are all the $i\times i$ reduced minors of $F,F_1$ and $F_2$ respectively, we have
  \[ J_i(F) =  \left\langle\frac{h_1^{(i)}}{d_i(F)},\ldots,
  \frac{h_{\beta_i}^{(i)}}{d_i(F)}\right\rangle, ~
  J_i(F_1) =  \left\langle\frac{f_1^{(i)}}{d_i(F_1)},\ldots,
  \frac{f_{\beta_i}^{(i)}}{d_i(F_1)}\right\rangle, ~
  J_i(F_2) =  \left\langle\frac{g_1^{(i)}}{d_i(F_2)},\ldots,
  \frac{g_{\beta_i}^{(i)}}{d_i(F_2)}\right\rangle.\]
  It follows from Equation \eqref{prop-equ-5} that
  \[J_i(F) \subseteq J_i(F_1) \subseteq K[x_1,\ldots,x_n] \text{ and } J_i(F) \subseteq J_i(F_2) \subseteq K[x_1,\ldots,x_n].\]
  Since $J_i(F) = K[x_1,\ldots,x_n]$, we have $J_i(F_1) = J_i(F_2) = K[x_1,\ldots,x_n]$, where $i=1,\ldots,l$.
 \end{proof}

 Lemma \ref{lemma-reduced-1} solves the second challenge raised in Example \ref{biv-example-idea}.

 \subsubsection{Solving the third challenge}

 \begin{lemma}\label{main-lem-2-multi}
  Let $F\in K[x_1,\ldots,x_n]^{l\times l}$, and $S_1,S_2\subset K[x_1]$ are comaximal multiplicatively closed. Then $F(x_1,\ldots,x_n) \sim_{K[x_1,\ldots,x_n]} F(x_1,0,\ldots,0)$ if and only if $F(x_1,\ldots,x_n) \sim_{(K[x_1]_{S_i})[x_2,\ldots,x_n]} F(x_1,0,\ldots,0)$, where $i=1,2$.
 \end{lemma}

 \begin{proof}
  As the necessity is trivial, we only need to prove the sufficiency.

  Since $F(x_1,\ldots,x_n) \sim_{(K[x_1]_{S_i})[x_2,\ldots,x_n]} F(x_1,0,\ldots,0)$, there exist two unimodular matrices $U_i$ and $V_i$ in $(K[x_1]_{S_i})[x_2,\ldots,x_n]^{l\times l}$ such that
  \begin{equation}\label{main-lem-equ-01-multi}
   F(x_1,\ldots,x_n) = U_i \cdot F(x_1,0,\ldots,0) \cdot V_i,
  \end{equation}
  where $i=1,2$. Setting the variable $x_n$ in Equation \eqref{main-lem-equ-01-multi} to zero, we have
  \begin{equation}\label{main-lem-equ-02-multi}
   F(x_1,\ldots,x_{n-1},0) = U_i(x_1,\ldots,x_{n-1},0) \cdot F(x_1,0,\ldots,0) \cdot V_i(x_1,\ldots,x_{n-1},0).
  \end{equation}
  It follows from Equation \eqref{main-lem-equ-02-multi} that
  \begin{equation}\label{main-lem-equ-03-multi}
   F(x_1,\ldots,x_{n-1},0) \sim_{(K[x_1]_{S_i})[x_2,\ldots,x_{n-1}]} F(x_1,0,\ldots,0).
  \end{equation}
  It is obvious that $(K[x_1]_{S_i})[x_2,\ldots,x_{n-1}] \subset (K[x_1]_{S_i})[x_2,\ldots,x_{n-1},x_n]$. This implies that
  \begin{equation*}\label{main-lem-equ-04-multi}
   F(x_1,\ldots,x_{n-1},x_n) \sim_{(K[x_1]_{S_i})[x_2,\ldots,x_n]} F(x_1,\ldots,x_{n-1},0).
  \end{equation*}
  According to Lemma \ref{main-lem-1-multi}, there are four unimodular matrices $C_1,D_1\in K[x_1,\ldots,x_n,z_1]^{l\times l}$, $C_2,$ $D_2\in K[x_1,\ldots,x_n,z_2]^{l\times l}$, and two polynomials $s_1\in S_1,s_2\in S_2$ such that
  \begin{gather}
   F(x_1,\ldots,x_{n-1},x_n+s_1z_1) = C_1(x_1,\ldots,x_n,z_1) \cdot F(x_1,\ldots,x_n) \cdot D_1(x_1,\ldots,x_n,z_1),\label{main-lem-equ-1-multi} \\
   F(x_1,\ldots,x_{n-1},x_n+s_2z_2) = C_2(x_1,\ldots,x_n,z_2) \cdot F(x_1,\ldots,x_n) \cdot D_2(x_1,\ldots,x_n,z_2),\label{main-lem-equ-2-multi}
  \end{gather}
  where $z_1,z_2$ are two new variables. Setting the variable $x_n$ in Equation \eqref{main-lem-equ-2-multi} to $x_n+s_1z_1$ and combining Equation \eqref{main-lem-equ-1-multi}, we obtain
  \begin{equation*}\label{main-lem-equ-4-multi}
   \begin{split}
   F(x_1,\ldots,x_{n-1},x_n+s_1z_1+s_2z_2)  ~ = & ~ ~  C_2(x_1,\ldots,x_{n-1},x_n+s_1z_1,z_2) \cdot C_1(x_1,\ldots,x_{n-1},x_n,z_1) \\ & ~ ~ \cdot F(x_1,\ldots,x_{n-1},x_n) \cdot D_1(x_1,\ldots,x_{n-1},x_n,z_1) \\ & ~ ~ \cdot D_2(x_1,\ldots,x_{n-1},x_n+s_1z_1,z_2).
   \end{split}
  \end{equation*}
  Let $U = C_2(x_1,\ldots,x_{n-1},x_n+s_1z_1,z_2) \cdot C_1(x_1,\ldots,x_{n-1},x_n,z_1)$ and $V = D_1(x_1,\ldots,x_{n-1},x_n,z_1) \cdot D_2(x_1,\ldots,x_{n-1},x_n+s_1z_1,z_2)$, then $U,V\in K[x_1,\ldots,x_{n-1},x_n,z_1,z_2]^{l\times l}$ are unimodular matrices. In addition, we have
  \begin{equation}\label{main-lem-equ-5-multi}
   F(x_1,\ldots,x_{n-1},x_n+s_1z_1+s_2z_2) = U \cdot F(x_1,\ldots,x_{n-1},x_n) \cdot V.
  \end{equation}
  Since $S_1,S_2$ are comaximal, there exist polynomials $u,v\in K[x_1]$ such that $us_1+vs_2=1$. Setting the variables $z_1$ and $z_2$ in Equation \eqref{main-lem-equ-5-multi} to $-ux_n$ and $-vx_n$ respectively, we have
  \begin{equation*}\label{main-lem-equ-6-multi}
   \begin{split}
   F(x_1,\ldots,x_{n-1},0) ~ = & ~ ~  U(x_1,\ldots,x_{n-1},x_n,-ux_n,-vx_n) \cdot F(x_1,\ldots,x_{n-1},x_n) \\ & ~ ~ \cdot V(x_1,\ldots,x_{n-1},x_n,-ux_n,-vx_n).
   \end{split}
  \end{equation*}
  It is obvious that $U(x_1,\ldots,x_{n-1},x_n,-ux_n,-vx_n)$ and $V(x_1,\ldots,x_{n-1},x_n,-ux_n,-vx_n)$ are two unimodular matrices in $K[x_1,\ldots,x_{n-1},x_n]^{l\times l}$. This implies that
  \begin{equation}\label{main-lem-equ-7-multi}
   F(x_1,\ldots,x_{n-1},x_n) \sim_{K[x_1,\ldots,x_{n-1},x_n]} F(x_1,\ldots,x_{n-1},0).
  \end{equation}
  Setting the variables $x_{n-1},x_n$ in Equation \eqref{main-lem-equ-01-multi} to zeros, we have
  \begin{equation}\label{main-lem-equ-8-multi}
   F(x_1,\ldots,x_{n-2},0,0) = U_i(x_1,\ldots,x_{n-2},0,0) \cdot F(x_1,0,\ldots,0) \cdot V_i(x_1,\ldots,x_{n-2},0,0).
  \end{equation}
  It follows from Equation \eqref{main-lem-equ-8-multi} that
  \begin{equation}\label{main-lem-equ-9-multi}
   F(x_1,\ldots,x_{n-2},0,0) \sim_{(K[x_1]_{S_i})[x_2,\ldots,x_{n-2}]} F(x_1,0,\ldots,0).
  \end{equation}
  Combining Equations \eqref{main-lem-equ-03-multi} and \eqref{main-lem-equ-9-multi}, we obtain
  \begin{equation}\label{main-lem-equ-10-multi}
   F(x_1,\ldots,x_{n-2},x_{n-1},0) \sim_{(K[x_1]_{S_i})[x_2,\ldots,x_{n-2},x_{n-1}]} F(x_1,\ldots,x_{n-2},0,0).
  \end{equation}
  Let $F_1 = F(x_1,\ldots,x_{n-1},0)$, then $F_1\in K[x_1,\ldots,x_{n-1}]^{l\times l}$. It follows from \eqref{main-lem-equ-10-multi} that
  \begin{equation*}\label{main-lem-equ-11-multi}
   F_1(x_1,\ldots,x_{n-2},x_{n-1}) \sim_{(K[x_1]_{S_i})[x_2,\ldots,x_{n-2},x_{n-1}]} F_1(x_1,\ldots,x_{n-2},0).
  \end{equation*}
  Using Lemma \ref{main-lem-1-multi} again, we obtain
  \begin{equation}\label{main-lem-equ-12-multi}
   F_1(x_1,\ldots,x_{n-1}) \sim_{K[x_1,\ldots,x_{n-1}]} F_1(x_1,\ldots,x_{n-2},0).
  \end{equation}
  It follows from $K[x_1,\ldots,x_{n-1}] \subset K[x_1,\ldots,x_{n-1},x_n]$ and Equations \eqref{main-lem-equ-7-multi} and \eqref{main-lem-equ-12-multi} that
  \begin{equation*}\label{main-lem-equ-13-multi}
   F(x_1,\ldots,x_n) \sim_{K[x_1,\ldots,x_n]} F(x_1,\ldots,x_{n-1},0) \sim_{K[x_1,\ldots,x_n]} F(x_1,\ldots,x_{n-2},0,0).
  \end{equation*}
  Repeating the above proof process a total of $n-1$ times, we can deduce that
  \begin{equation*}\label{main-lem-equ-14-multi}
   F(x_1,\ldots,x_n) \sim_{K[x_1,\ldots,x_n]} F(x_1,0,\ldots,0).
  \end{equation*}
  The proof is completed.
 \end{proof}

 \begin{lemma}\label{biv-main-lemma-1-multi}
  Let $f_1,\ldots,f_l,g_1,\ldots,g_l\in K[x_1]\setminus \{0\}$ satisfy $f_i\mid f_{i+1}$ and $g_i\mid g_{i+1}$ for $i=1,\ldots,l-1$, and $\gcd(f_l,g_l)=1$. If $U \in K[x_1,\ldots,x_n]^{l\times l}$ is a unimodular matrix, then
  \[ \diag\{f_1,\ldots,f_l\} \cdot U \cdot \diag\{g_1,\ldots,g_l\}
     \sim_{K[x_1,\ldots,x_n]} \diag\{f_1g_1,\ldots,f_lg_l\}.\]
 \end{lemma}

 \begin{proof}
  We construct the following two sets
  \[ S_1 = \left\{f_l^{t_{1}} \mid t_{1}\in \mathbb{N}\right\} \text{ and }
  S_2 = \left\{g_l^{t_{2}} \mid t_{2}\in \mathbb{N}\right\}.\]
  It is straightforward to show that $S_1,S_2 \subset K[x_1]$ are multiplicatively closed. Moreover, it follows from $\gcd(f_l,g_l)=1$ that $S_1,S_2$ are comaximal.

  Let $A = \diag\{f_1,\ldots,f_l\} \cdot U \cdot \diag\{g_1,\ldots,g_l\}$ and $U = \left( u_{ij} \right)_{l\times l}$. On the one hand, we have
  \begin{equation}\label{main-lem-equ-3-1-multi}
   A = \left(
         \begin{array}{cccc}
           u_{11} & \frac{f_1}{f_2}\cdot u_{12} & \cdots  & \frac{f_1}{f_l} \cdot u_{1l} \\  \frac{f_2}{f_1}\cdot  u_{21} &  u_{22} & \cdots & \frac{f_2}{f_l}\cdot u_{2l} \\
           \vdots & \vdots & \ddots & \vdots \\
           \frac{f_l}{f_1}\cdot u_{l1} & \frac{f_l}{f_2}\cdot u_{l2} & \cdots & u_{ll} \\
         \end{array}
       \right) \cdot \diag\{f_1g_1,\ldots,f_lg_l\}.
  \end{equation}
  This implies that
  \begin{equation}\label{main-lem-equ-3-2-multi}
   A(x_1,\ldots,x_n) \sim_{(K[x_1]_{S_1})[x_2,\ldots,x_n]} \diag\{f_1g_1,\ldots,f_lg_l\}.
  \end{equation}
  Setting the variables $x_2,\ldots,x_n$ in Equation \eqref{main-lem-equ-3-1-multi} to zeros, and we have
  \begin{equation}\label{main-lem-equ-3-3-multi}
   A(x_1,0,\ldots,0) \sim_{(K[x_1]_{S_1})[x_2,\ldots,x_n]} \diag\{f_1g_1,\ldots,f_lg_l\}.
  \end{equation}
  Combining Equations \eqref{main-lem-equ-3-2-multi} and \eqref{main-lem-equ-3-3-multi}, we get
  \begin{equation*}\label{main-lem-equ-3-4-multi}
   A(x_1,\ldots,x_n)  \sim_{(K[x_1]_{S_1})[x_2,\ldots,x_n]} A(x_1,0,\ldots,0).
  \end{equation*}
  On the other hand, we have
  \begin{equation*}\label{main-lem-equ-3-5-multi}
   A = \diag\{f_1g_1,\ldots,f_lg_l\} \cdot \left(
         \begin{array}{cccc}
           u_{11} & \frac{g_2}{g_1}\cdot u_{12} & \cdots  & \frac{g_l}{g_1} \cdot u_{1l} \\  \frac{g_1}{g_2}\cdot  u_{21} &  u_{22} & \cdots & \frac{g_l}{g_2}\cdot u_{2l} \\
           \vdots & \vdots & \ddots & \vdots \\
           \frac{g_1}{g_l}\cdot u_{l1} & \frac{g_2}{g_l}\cdot u_{l2} & \cdots & u_{ll} \\
         \end{array}
       \right).
  \end{equation*}
  An argument similar to the above proof shows that
  \begin{equation*}\label{main-lem-equ-3-6-multi}
   A(x_1,\ldots,x_n)  \sim_{(K[x_1]_{S_2})[x_2,\ldots,x_n]} A(x_1,0,\ldots,0).
  \end{equation*}
  According to Lemma \ref{main-lem-2-multi}, we get
  \begin{equation}\label{main-lem-equ-3-7-multi}
   A(x_1,\ldots,x_n)  \sim_{K[x_1,\ldots,x_n]} A(x_1,0,\ldots,0).
  \end{equation}
  In the following we prove that the Smith form of $A(x_1,0,\ldots,0)$ is $\diag\{f_1g_1,\ldots,f_lg_l\}$. Obviously,
  \[A(x_1,0,\ldots,0) = \diag\{f_1,\ldots,f_l\} \cdot U(x_1,0,\ldots,0) \cdot \diag\{g_1,\ldots,g_l\}.\]
  Since $U \in K[x_1,\ldots,x_n]^{l\times l}$ is a unimodular matrix, $\det(U)$ is a unit in $k$. Without loss of generality, assume that $\det(U) = 1$. Then, we have $\det(U(x_1,0,\ldots,0)) = 1$. Let $B(x_1) = \diag\{f_1,\ldots,f_l\} \cdot U(x_1,0,\ldots,0)$, then
  \[B(x_1) \sim_{K[x_1]} \diag\{f_1,\ldots,f_l\}.\]
  It follows from Proposition \ref{lemma-reduced-2} that
  \[d_i(B(x_1)) = f_1\cdots f_i, ~ i=1,\ldots,l.\]
  Since $A(x_1,0,\ldots,0)= B(x_1) \cdot \diag\{g_1,\ldots,g_l\}$, by Lemma \ref{lemma-reduced-1} we have
  \[ d_i(A(x_1,0,\ldots,0)) = (f_1g_1)\cdots (f_ig_i), ~ i=1,\ldots,l.\]
  It follows that the Smith form of $A(x_1,0)$ is $\diag\{f_1g_1,\ldots,f_lg_l\}$. By the fact that $K[x_1]$ is a principal ideal domain, we obtain
  \begin{equation}\label{main-lem-equ-3-8-multi}
   A(x_1,0,\ldots,0)  \sim_{K[x_1]} \diag\{f_1g_1,\ldots,f_lg_l\}.
  \end{equation}
  Combining Equations \eqref{main-lem-equ-3-7-multi} and \eqref{main-lem-equ-3-8-multi}, and $K[x_1] \subset K[x_1,\ldots,x_n]$, we infer that
  \begin{equation*}\label{main-lem-equ-3-9-multi}
   A(x_1,\ldots,x_n)  \sim_{K[x_1,\ldots,x_n]} \diag\{f_1g_1,\ldots,f_lg_l\}.
  \end{equation*}
  The proof is completed.
 \end{proof}

 Lemma \ref{biv-main-lemma-1-multi} solves the third and most crucial challenge raised in Example \ref{biv-example-idea}.

\subsection{Main results}

 Now, we present the first important result in this paper.

 \begin{theorem}\label{main-theorem-1}
  Let $F\in K[x_1,\ldots,x_n]^{l\times l}$ and $\det(F)\in K[x_1]$, then $F$ is equivalent to its Smith form if and only if $J_i(F) = K[x_1,\ldots,x_n]$ for $i=1,\ldots,l$.
 \end{theorem}

 \begin{proof}
 {\em Necessity}. Let $H=\diag\{h_1,\ldots,h_l\} \in K[x_1,\ldots,x_n]^{l\times l}$ be the Smith form of $F$, then $J_i(H) = K[x_1,\ldots,x_n]$ for $i=1,\ldots,l$. If $F$ is equivalent to $H$, then by Proposition \ref{lemma-reduced-2} we have $J_i(F) = J_i(H)$ for $i=1,\ldots,l$.

  {\em Sufficiency}. Let $\det(F)=p_1^{r_1}p_2^{r_2}\cdots p_t^{r_t}$, where $p_1,p_2,\ldots,p_t\in k[x_1]$ are distinct and irreducible polynomials, and $r_1,r_2,\ldots,r_t$ are positive integers. In the following, we will complete the proof by induction on $t$. When $t=1$, under the condition that $J_i(F)= k[x_1,\ldots,x_n]$ for $i=1,\ldots,l$, $F$ is equivalent to its Smith form by Remark \ref{remark_lu-1}. Now assume that the proposition ``if $J_i(F)= k[x_1,\ldots,x_n]$ for $i=1,\ldots,l$, then $F$ is equivalent to its Smith form'' holds true for $t \leq N-1$. When $t=N$, by repeatedly applying Lemma \ref{Lemma_chal_1}, we can obtain two polynomial matrices $F_1,F_2\in k[x_1,\ldots,x_n]^{l\times l}$ such that
  \[ F = F_1 \cdot F_2, \text{ where } \det(F_1) = p_1^{r_1} \text{ and } \det(F_2) = p_2^{r_2}\cdots p_N^{r_N}.\]
  Based on Lemma \ref{lemma-reduced-1}, $J_i(F_1)= J_i(F_2)= k[x_1,\ldots,x_n]$, where $i=1,\ldots,l$. Let $\diag\{\Phi_1,\ldots,\Phi_l\}$ and $\diag\{\Psi_1,\ldots,\Psi_l\}$ be the Smith forms of $F_1$ and $F_2$ respectively, where $\Phi_j\mid \Phi_{j+1}$ and $\Psi_j\mid \Psi_{j+1}$, $j=1,\ldots,l-1$. By inductive hypothesis, we have
  \[ F_1 \sim_{k[x_1,\ldots,x_n]} \diag\{\Phi_1,\ldots,\Phi_l\} \text{ and }
     F_2 \sim_{k[x_1,\ldots,x_n]} \diag\{\Psi_1,\ldots,\Psi_l\}.\]
  It follows that
  \[ F \sim_{k[x_1,\ldots,x_n]} \diag\{\Phi_1,\ldots,\Phi_l\} \cdot U \cdot \diag\{\Psi_1,\ldots,\Psi_l\},\]
  where $U\in k[x_1,\ldots,x_n]^{l\times l}$ is a unimodular matrix. According to Lemma \ref{biv-main-lemma-1-multi}, we obtain
  \[ F \sim_{k[x_1,\ldots,x_n]} \diag\{\Phi_1\Psi_1,\ldots,\Phi_l\Psi_l\}.\]
  It follows from Proposition \ref{lemma-reduced-2} that
  \[ d_i(F) = (\Phi_1\Psi_1)\cdots (\Phi_i\Psi_i), ~ i=1,\ldots,l.\]
  This implies that $\diag\{\Phi_1\Psi_1,\ldots,\Phi_l\Psi_l\}$ is the Smith form of $F$. Therefore, the above proposition holds for $t=N$.
 \end{proof}

 Next, we generalize Theorem \ref{main-theorem-1} to the case of non-square and non-full rank matrices.

 \begin{corollary}\label{main-corollary}
  Let $F\in K[x_1,\ldots,x_n]^{l\times m}$ with rank $r$ and $d_r(F)\in K[x_1]$, where $1 \leq r \leq l$. Then $F$ is equivalent to its Smith form if and only if $J_i(F) = K[x_1,\ldots,x_n]$ for $i=1,\ldots,r$.
 \end{corollary}

 \begin{proof}
  Since the necessity is obvious, we only need to prove the sufficiency.

  Since $J_r(F)=K[x_1,\ldots,x_n]$, by Theorem \ref{Lin-Bose-theorem} there are two polynomial matrices $G_1\in K[x_1,$ $\ldots,x_n]^{l\times r}$ and $F_1\in K[x_1,\ldots,x_n]^{r\times m}$ such that
  \begin{equation}\label{coro-equ-1}
   F = G_1\cdot F_1,
  \end{equation}
  where $F_1$ is a ZLP matrix. According to the Quillen-Suslin theorem, there exits a unimodular matrix $U\in K[x_1,\ldots,x_n]^{m\times m}$ such that
  \begin{equation}\label{coro-equ-2}
   F_1\cdot U = (I_r, 0_{r\times (m-r)}),
  \end{equation}
  where $I_r$ is the $r\times r$ identity matrix. Multiplying both sides of Equation \eqref{coro-equ-1} by $U$ and combining it with Equation \eqref{coro-equ-2}, we get
  \begin{equation}\label{coro-equ-3}
   F \cdot U = (G_1, 0_{r\times (m-r)}).
  \end{equation}
  It follows that $F\sim_{K[x_1,\ldots,x_n]} (G_1, 0_{r\times (m-r)})$. Based on Proposition \ref{lemma-reduced-2}, we have $J_r(G_1) = J_r(F) =  K[x_1,\ldots,x_n]$. Using Theorem \ref{Lin-Bose-theorem} again, there are two matrices $G_2\in K[x_1,\ldots,x_n]^{l\times r}$ and $G_3 \in K[x_1,\ldots,x_n]^{r\times r}$ such that $G_1=G_2\cdot G_3$, where $G_2$ is a zero right prime matrix. It follows from the Quillen-Suslin theorem that there is a unimodular matrix $V\in K[x_1,\ldots,x_n]^{l\times l}$ such that
  \begin{equation}\label{coro-equ-5}
   V \cdot G_2 = \left(\begin{array}{c} I_r \\
   0_{(l-r)\times r} \\ \end{array}\right).
  \end{equation}
  Multiplying both sides of Equation \eqref{coro-equ-3} by $V$ and combining it with Equation \eqref{coro-equ-5}, we have
  \begin{equation}\label{coro-equ-6}
   V \cdot F \cdot U = \left(\begin{array}{cc} G_3 & 0_{r\times (m-r)} \\
   0_{(l-r)\times r} & 0_{(l-r)\times (m-r)} \\ \end{array}\right).
  \end{equation}
  By Proposition \ref{lemma-reduced-2}, we obtain
  \begin{equation}\label{coro-equ-4}
   d_i(G_3) = d_i(F) \text{ and } J_i(G_3) = J_i(F) = K[x_1,\ldots,x_n],
  \end{equation}
  where $i=1,\ldots,r$. Let $D\in K[x_1,\ldots,x_n]^{r\times r}$ be the Smith form of $G_3$, then $G_3\sim_{K[x_1,\ldots,x_n]} D$ by Theorem \ref{main-theorem-1}. Therefore, we have
  \begin{equation*}\label{coro-equ-7}
   F \sim_{K[x_1,\ldots,x_n]} \left(\begin{array}{cc} D & 0_{r\times (m-r)} \\
   0_{(l-r)\times r} & 0_{(l-r)\times (m-r)} \\ \end{array}\right).
  \end{equation*}
  This implies that $F$ is equivalent to its Smith form.
 \end{proof}

\section{Concluding remarks}\label{sec_conclusions}

 Let $F\in K[x_1,\ldots,x_n]^{l\times m}$ with rank $r$. In this paper we study the equivalence problem of $F$ and its Smith form, where $F$ satisfies the condition that $d_r(F)\in K[x_1]$. Through an illustrative example, we analyzed three key challenges that arose during the process of equivalence and successfully resolved them one by one. Compared with previous works, one of the focuses of this paper is to prove the validity of Lemma \ref{biv-main-lemma-1-multi}. With the help of the localization technique, we have ultimately obtained the necessary and sufficient condition for $F$ to be equivalent to its Smith form, which is all the $i\times i$ reduced minors of $F$ generate $K[x_1,\ldots,x_n]$, where $i=1,\ldots,r$. It is evident that we have generalized Frost and Storey's assertion to a wider scope. We hope the results provided in the paper will motivate further research in the area of the equivalence of multivariate polynomial matrices.

\section*{Acknowledgments}

 This research was supported by the National Natural Science Foundation of China under Grant Nos. 12171469 and 12201210, the Natural Science Foundation of Sichuan Province under Grant No. 2024NSFSC0418, the National Key Research and Development Program under Grant No. 2020YFA0712300, and the Fundamental Research Funds for the Central Universities under Grant No. 2682024ZTPY052.

\begin{appendix}

\section{\label{sec:appendix}}

 Although the proof of Lemma \ref{main-lem-1-multi} is similar to that of Lemma \ref{main-lem-yengui}, for the sake of the rigor of the argument and the ease of understanding we still give a detailed proof here.

 \begin{proof}[Proof of Lemma {\rm \ref{main-lem-1-multi}}]
  Since $F(x_1,\ldots,x_{n-1},x_n) \sim_{(K[x_1]_S)[x_2,\ldots,x_n]} F(x_1,\ldots,x_{n-1},0)$, there are two unimodular matrices $U,V\in (K[x_1]_S)[x_2,\ldots,x_n]^{l\times l}$ such that
  \begin{equation}\label{main-equ-1-multi}
   F(x_1,\ldots,x_{n-1},x_n) = U\cdot F(x_1,\ldots,x_{n-1},0) \cdot V.
  \end{equation}
  Setting the variable $x_n$ in Equation \eqref{main-equ-1-multi} to $x_n+z$, we get
 \begin{equation}\label{main-equ-2-multi}
  \begin{split}
   F(x_1,\ldots,x_{n-1},x_n+z) ~ = & ~ ~ U(x_1,\ldots,x_{n-1},x_n+z)\cdot U^{-1}\cdot F(x_1,\ldots,x_{n-1},x_n) \\
       &  ~ ~ \cdot V^{-1}\cdot V(x_1,\ldots,x_{n-1},x_n+z).
  \end{split}
 \end{equation}

  Let $E(x_1,\ldots,x_{n-1},x_n,z) = U(x_1,\ldots,x_{n-1},x_n+z)\cdot U^{-1}$, then $\det(E)=1$ by the fact that $\det(U)$ is a unit in $K[x_1]_S$. Expanding $E(x_1,\ldots,x_{n-1},x_n,z)$ into a matrix polynomial w.r.t. $z$, then
  \begin{equation*}\label{main-equ-3-multi}
   \begin{split}
    E(x_1,\ldots,x_{n-1},x_n,z) ~ = & ~ ~ E_0(x_1,\ldots,x_{n-1},x_n) + E_1(x_1,\ldots,x_{n-1},x_n)\cdot z \\
    & ~ ~ + E_2(x_1,\ldots,x_{n-1},x_n)\cdot z^2 + \cdots + E_d(x_1,\ldots,x_{n-1},x_n)\cdot z^d,
   \end{split}
  \end{equation*}
  where $E_i(x_1,\ldots,x_{n-1},x_n)\in (K[x_1]_S)[x_2,\ldots,x_n]^{l\times l}$ and $i=0,1,\ldots,d$. Since $E(x_1,\ldots,x_{n-1},x_n,$ $0) = U\cdot U^{-1} = I_l$, we have $E_0(x_1,\ldots,x_{n-1},x_n) = I_l$. Moreover, for each $i$ there exists $s_i\in S$ such that
  \begin{equation*}\label{main-equ-4-multi}
   E_i(x_1,\ldots,x_{n-1},x_n)\cdot (s_iz)^i\in K[x_1,\ldots,x_{n-1},x_n,z]^{l\times l}.
  \end{equation*}
  Let $s'=\lcm(s_1,\ldots,s_d)$, then $E(x_1,\ldots,x_{n-1},x_n,s'z)\in K[x_1,\ldots,x_{n-1},x_n,z]^{l\times l}$. It is obvious that $\det(E(x_1,\ldots,x_{n-1},x_n,s'z))=1$. This implies that $E(x_1,\ldots,x_{n-1},x_n,s'z)$ is a unimodular matrix.

  Let $H(x_1,\ldots,x_{n-1},x_n,z) = V^{-1}\cdot V(x_1,\ldots,x_{n-1},x_n+z)$, then $\det(H)=1$ by the fact that $\det(V)$ is a unit in $K[x_1]_S$. Proceeding as the above proof, there is another polynomial $s''\in S$ such that $H(x_1,\ldots,x_{n-1},x_n,s''z)\in K[x_1,\ldots,x_{n-1},x_n,z]^{l\times l}$ and is a unimodular matrix.

  Let $s = \lcm(s',s'')$ and substitute $sz$ for $z$ in Equation \eqref{main-equ-2-multi}, we obtain
  \begin{equation*}\label{main-equ-5-multi}
   F(x_1,\ldots,x_{n-1},x_n+sz) = E(x_1,\ldots,x_{n-1},x_n,sz)\cdot F(x_1,\ldots,x_{n-1},x_n)\cdot H(x_1,\ldots,x_{n-1},x_n,sz),
  \end{equation*}
  where $E(x_1,\ldots,x_{n-1},x_n,sz),H(x_1,\ldots,x_{n-1},x_n,sz)\in K[x_1,\ldots,x_{n-1},x_n,z]^{l\times l}$ are unimodular matrices. By the fact that $S$ is multiplicatively closed, it follows that $s\in S$. Let $C(x_1,\ldots,x_{n-1},x_n,$ $z) = E(x_1,\ldots,x_{n-1},x_n,sz)$ and $D(x_1,\ldots,x_{n-1},x_n,z) = H(x_1,\ldots,x_{n-1},x_n,sz)$, then
  \begin{equation*}\label{main-equ-6-multi}
   F(x_1,\ldots,x_{n-1},x_n+sz) = C(x_1,\ldots,x_{n-1},x_n,z) \cdot F(x_1,\ldots,x_{n-1},x_n) \cdot D(x_1,\ldots,x_{n-1},x_n,z).
  \end{equation*}
  The proof is completed.
 \end{proof}

\end{appendix}

\end{document}